\documentclass[11pt,a4paper]{article}
\usepackage{amsmath,amssymb,amsthm}
\usepackage{graphicx}
\usepackage[margin=2.5cm]{geometry}
\usepackage[usenames,dvipsnames]{xcolor}
\usepackage[colorlinks=true,linkcolor=blue,citecolor=blue]{hyperref}

\newif\ifnewmaterial
\newmaterialfalse
\ifnewmaterial
  \long\def\new#1{{\color{ForestGreen}#1}}
\else
  \long\def\new#1{#1}
\fi

\newtheorem{theorem}{Theorem}
\newtheorem{proposition}[theorem]{Proposition}
\newtheorem{corollary}[theorem]{Corollary}
\newtheorem{lemma}[theorem]{Lemma}
\theoremstyle{remark}

\newcommand{\ff}[2]{\left(#1\right)_{#2}}

\title{\bf Static anisotropic stars in Lovelock gravity: Universal closed-form equations}

\author{Sudan Hansraj\\
\small Astrophysics Research Centre, School of Mathematics, Statistics and
Computer Science,\\
\small University of KwaZulu--Natal, Private Bag X54001, Durban 4000,
South Africa}

\date{\today}

\begin{document}
\maketitle

\begin{abstract}

A characteristic-polynomial formulation is already known for static perfect fluids in Lovelock gravity [Phys. Rev. D 113, 084008 (2026)]. We extend it to the complete anisotropic matter sector. For arbitrary dimension and independently coupled curvature orders, we derive closed, summation-free expressions for the density, radial pressure, tangential pressure and anisotropy in terms of a single polynomial \(W\), its derivatives and two derived polynomials \(V\) and \(U\). The dimension-dependent combinatorial coefficients are obtained directly from the generalized-delta antisymmetrisation and independently confirmed through an inductive argument. Imposing pressure isotropy recovers the known total-derivative equation, which we resolve into an explicit operator affine in Lovelock order and linear in dimension. This structure proves that, apart from the branch of vanishing angular sectional curvature, the interior Schwarzschild geometry is the unique isotropic interior common to arbitrary Lovelock couplings. It also identifies the fixed-sign obstruction to bounded pure Lovelock spheres in \(d=2N+1\) with the polynomial identity \(V\equiv0\), and shows how a lower-order term or cosmological term removes this obstruction. Finally, we integrate the resulting arbitrary-order stellar equations with all admissible orders active in \(d=10\) and \(d=11\). For the positive single-parameter hierarchy of Lovelock couplings considered, the general-relativistic sequence topology persists while higher orders increase the maximum mass cumulatively.

\end{abstract}

\section{Introduction}

Lovelock gravity \cite{Lovelock1971,Lovelock1972} is the unique generalisation of general relativity in dimensions $d>4$ whose field
equations remain of second order and whose field tensor is
divergence-free. Its Lagrangian is a sum of dimensionally continued Euler
densities, each carrying its own coupling, and written out term by term. The resulting equations become  unmanageable quickly for example the third-order contribution to a static sphere already occupies several lines.

The Lovelock construction does have a long history. The quadratic combination now known as the Gauss--Bonnet density was singled out in four dimensions by Lanczos \cite{Lanczos1938}, where its integral is topological and the local dynamics therefore unaffected; what Lovelock established is that this is the first member of a finite hierarchy, and that the hierarchy is exhausted. The identification of these densities as dimensional
continuations of the Euler characteristic, and the special role of the critical dimensions in which an individual density becomes topological or kinematic, have been developed from several different points of view
\cite{Zumino1986,TeitelboimZanelli1987,PadmanabhanKothawala2013,Charmousis2009}.

It should be remembered that higher-curvature corrections are not motivated by geometry alone. Zwiebach showed that the Gauss--Bonnet combination can occur in the low-energy
string effective action while avoiding the ghost mode that accompanies a generic quadratic correction \cite{Zwiebach1985}, and related higher-dimensional constructions followed shortly afterwards 
\cite{Zumino1986,Myers1987}. This supplied much of the early physical motivation for Einstein--Gauss--Bonnet gravity and, more broadly, for treating the Lovelock hierarchy as a controlled model of higher-curvature
dynamics; the theories have since served as laboratories for black-hole physics, dimensional reduction, braneworld gravity and higher-dimensional thermodynamics \cite{BanadosTeitelboimZanelli1994,%
CrisostomoTroncosoZanelli2000,GarraffoGiribet2008,%
PadmanabhanKothawala2013}.

Static spherical solutions have consequently been central to understanding the theory. In this symmetry class, the separate Lovelock orders combine into a single curvature polynomial. In vacuum, Wheeler  \cite{Wheeler1986} showed that the contributions from all active Lovelock orders combine into a single algebraic polynomial for the angular sectional curvature \(\psi\). The Einstein part is the Schwarzschild--Tangherlini solution \cite{Tangherlini1963} and the Gauss--Bonnet contribution the Boulware--Deser solution \cite{BoulwareDeser1985}; the general structure
was developed in \cite{Whitt1988,MyersSimon1988,Cai2002,Zegers2005}, and
the dimensionally continued and pure Lovelock black holes exhibited further universal features related to the critical odd and even dimensions
\cite{TeitelboimZanelli1987,BanadosTeitelboimZanelli1994,%
DadhichPonsPrabhu2012,DadhichPonsPrabhu2013}.

Fluid interiors constitute a more stringent test, since matter couples simultaneously to the algebraic and to the derivative parts of the curvature polynomial. Studies of uniform density, isothermal spheres, exact compact-star models and compactness bounds have established a
number of order-independent and dimension-dependent properties in pure Lovelock and Einstein--Gauss--Bonnet gravity
\cite{DadhichMolinaKhugaev2010,DadhichHansrajMaharaj2016,%
DadhichHansrajChilambwe2017,DadhichChakraborty2017,%
ChakrabortyDadhich2020}
At Gauss--Bonnet order, the construction of isotropic interiors already presents a substantial integrability problem.  When the spatial potential  is prescribed, pressure isotropy is a second-order linear equation for the temporally directed potential; when the temporal potential is prescribed instead, the master isotropy equation becomes a nonlinear first-order equation for the spatially directed potential, commonly of Abel type.  Exact five-dimensional models were obtained through special potential choices, Frobenius expansions and barotropic reductions \cite{HansrajChilambweMaharaj2015,ChilambweHansrajMaharaj2015,MaharajChilambweHansraj2015}.  Subsequent work developed constant-potential and six-dimensional solutions, conformally flat families, solution-generating transformations and curvature-coordinate constructions \cite{HansrajMaharajChilambwe2019,HansrajMkhize2020,HansrajGovenderBanerjeeMkhize2021,MaharajHansrajSahoo2021,HansrajKrupanandanBanerjeeHansraj2022}.

For a static perfect fluid, moreover, a
substantial polynomial reduction has  already become  available in the literature. Recent work on
quasi-topological gravities has shown that the two independent field equations, the equation of hydrostatic equilibrium, the second-order equation governing the lapse and the constant density problem can all be written through a characteristic function of the sectional curvature of the orbit spheres and its derivatives \cite{Bueno2026}. Since the spherically symmetric equations of Lovelock and quasi-topological
theories coincide up to the dimension-dependent truncation of the series, those results already cover the perfect-fluid Lovelock sector, and we recover several of them below as checks rather than as new statements.

For a \emph{pure} Lovelock theory, in which a single curvature order is active, the transverse pressure and the isotropy condition at arbitrary order are likewise known, obtained from the conservation equation in
\cite{DadhichHansrajMaharaj2016}. What has not been made explicit is the corresponding statement for the full hierarchy, with every admissible order simultaneously active and independently coupled, and before pressure isotropy is imposed. A perfect fluid is an assumption on the source, under which the angular field equation is redundant by the Bianchi identity and never has to be written down; to our knowledge no summation-free characteristic-polynomial expression has been given for
the anisotropy of the full independently coupled Lovelock hierarchy. The present paper supplies the resummation of the independently coupled hierarchy into the closed anisotropic form, derives the dimension-dependent weights directly from the antisymmetrisation rather than assuming them, and resolves the isotropy condition into
contributions of definite curvature order.

In this work we  carry out  the reduction through the entire matter sector. Three combinations of the metric functions, written below as $a$, $b$ and $L$, turn out to exhaust the geometry that the matter equations are infleunced by, and the couplings enter only through a characteristic polynomial $W(\psi)$ in the sectional curvature $\psi$ of the orbit spheres, together with the two polynomials $V$ and $U$ obtained from it by applying a single differential operation twice. The density, the radial pressure and the tangential pressure of an arbitrary Lovelock theory in an arbitrary dimension are then given by \begin{equation*}
  2\rho=V-bW_{\psi},
  \qquad
  2p_{r}=aW_{\psi}-V,
  \qquad
  2(d-2)p_{t}=LW_{\psi}+(a+b)V_{\psi}-abW_{\psi\psi}-U 
\end{equation*}
so that there is  no summation over the curvature order. The whole hierarchy has been absorbed into one polynomial, and the dimension enters only through integer coefficients. This is made precise in \eqref{eq:closed}, the principal result of the paper, of which the tangential pressure and the anisotropy are the parts not previously known. The pressure isotropy condition collapses in the same way, to a single total derivative of the
first derivative of the characteristic polynomial multiplied by the metric functions; that equation, \eqref{eq:totalderiv}, recovers the  perfect fluid form  already reported in \cite{Bueno2026}. What is novel alongside \eqref{eq:closed} is its resolution into contributions of definite curvature order. Everything else in the paper either derives these or draws consequences from them.

This article is organised as follows.
Section~\ref{sec:weights} derives the dimension-dependent weights twice over: first by a direct contraction of the generalised Kronecker delta over the angular directions through a combinatorial argument, and then, as an independent confirmation, by establishing recurrences in the curvature order through an induction process. Section~\ref{sec:rhopr}
records the density and radial pressure at each order and cumulatively.
Section~\ref{sec:pt} obtains the tangential pressure from conservation
rather than from a separate projection, and displays the cancellation
which keeps the system linear in the temporal potential.
Sections~\ref{sec:resum} and \ref{sec:Ind} perform the resummation
leading to \eqref{eq:closed} and \eqref{eq:totalderiv}, and give a
further reduction: for a single Euler density the isotropy condition
becomes an operator $\mathcal I_{n,d}[Z,y]$ depending on the theory only
through the integers $n$ and $d$. Section~\ref{sec:checks} recovers the
five- and six-dimensional Einstein--Gauss--Bonnet equations and records
the seven-dimensional cubic case, comparing it with the third-order
equations obtained independently in the literature.
Sections~\ref{sec:uniform} and \ref{sec:uniqueness} treat uniform
density, showing first that the interior Schwarzschild metric solves the
isotropy condition in every Lovelock theory and then that, apart from one
exceptional branch, no other interior does.
Section~\ref{sec:critical} classifies the critical dimensions and
Section~\ref{sec:nogo} identifies the mechanism behind the known
non-existence of bounded pure Lovelock spheres, exhibiting two minimal
departures which remove it. Section~\ref{sec:sequences} integrates the
resulting stellar structure equations with every admissible curvature
order active in ten and eleven dimensions.

\section{Geometry and variables}
\label{sec:blocks}

We take the static spherically symmetric line element in the form 
\begin{equation}
  ds^{2}=e^{\nu(r)}dt^{2}-e^{\lambda(r)}dr^{2}-r^{2}d\Omega_{d-2}^{2},
\label{eq:metric}
\end{equation}
and introduce the Buchdahl variables
\begin{equation}
  x=Cr^{2},\qquad Z(x)=e^{-\lambda(r)},\qquad e^{\nu(r)}=y^{2}(x),
  \qquad \dot{\ }=\frac{d}{dx},
\label{eq:xZy}
\end{equation}
with $C$ a non-zero constant; a prime denotes $d/dr$. For purposes of simplicity we write the
geometric quantities
\begin{equation}
  \psi=\frac{1-Z}{r^{2}}=C\,\frac{1-Z}{x},\qquad
  a=\frac{Z\nu'}{r}=4CZ\frac{\dot y}{y},\qquad
  b=\frac{Z'}{r}=2C\dot Z,
\label{eq:psiab}
\end{equation}
together with
\begin{equation}
  L=\frac{2\sqrt Z}{y}\bigl(\sqrt Z\,y'\bigr)'
   =4C\left(2xZ\frac{\ddot y}{y}+Z\frac{\dot y}{y}
    +x\dot Z\frac{\dot y}{y}\right),
\label{eq:L}
\end{equation}
$\psi$ being the sectional curvature of the orbit spheres. We take
$T^{a}{}_{b}=\operatorname{diag}(\rho,-p_{r},-p_{t},\ldots,-p_{t})$ with
the convention $\sum_{n}\alpha_{n}G^{(n)a}{}_{b}=-T^{a}{}_{b}$.

For the metric \eqref{eq:metric} the Riemann tensor has only four
independent blocks,
\begin{equation}
  R^{tr}{}_{tr}=-\tfrac12 L,\qquad
  R^{ti}{}_{tj}=-\tfrac12 a\,\delta^{i}{}_{j},\qquad
  R^{ri}{}_{rj}=-\tfrac12 b\,\delta^{i}{}_{j},\qquad
  R^{ij}{}_{kl}=\psi\,\delta^{ij}_{kl},
\label{eq:blocks}
\end{equation}
with $i,j,k,l$ running over the $d-2$ angular directions and
$\delta^{ij}_{kl}=\delta^{i}{}_{k}\delta^{j}{}_{l}
-\delta^{i}{}_{l}\delta^{j}{}_{k}$. The last entry is the origin of
everything that follows: the orbit-sphere block is of constant curvature,
that is, it is a generalised Kronecker delta multiplied by $\psi$. Since
the Lovelock construction consists of antisymmetrised products of Riemann
tensors, and the generalised delta is the identity of that operation, the
sphere block can only return powers of $\psi$ times combinatorial factors.

\section{Combinatorial weights: direct count and inductive confirmation}
\label{sec:weights}

To commence let us write  the Lagrangian as $\mathcal L=\sum_{n}\alpha_{n}\mathcal L_{n}$
with
$\mathcal L_{n}=2^{-n}\delta^{c_{1}d_{1}\cdots c_{n}d_{n}}
_{e_{1}f_{1}\cdots e_{n}f_{n}}\prod_{i}R^{e_{i}f_{i}}{}_{c_{i}d_{i}}$,
so that $\mathcal L_{1}$ is the Ricci scalar and $\mathcal L_{2}$ the
Gauss--Bonnet invariant, with field tensor
\begin{equation}
  G^{(n)a}{}_{b}=-\frac{1}{2^{\,n+1}}\,
  \delta^{a\,c_{1}d_{1}\cdots c_{n}d_{n}}_{b\,e_{1}f_{1}\cdots e_{n}f_{n}}
  \prod_{i=1}^{n}R^{e_{i}f_{i}}{}_{c_{i}d_{i}},
\label{eq:lovelocktensor}
\end{equation}
$G^{(1)a}{}_{b}$ being the Einstein tensor. Throughout we use the falling
factorial
\begin{equation}
  \ff{m}{p}=m(m-1)\cdots(m-p+1),
\label{eq:fallingfactorial}
\end{equation}
with $\ff{m}{0}=1$; for non-negative integer $m$ it vanishes identically
when $p>m$. The single contraction identity we require is that of a
rank-$p$ generalised delta over the angular directions,
\begin{equation}
  \delta^{i_{1}\cdots i_{p}}_{i_{1}\cdots i_{p}}=\ff{d-2}{p}.
\label{eq:deltacontraction}
\end{equation}

\subsection{Direct combinatorial derivation}
\label{sec:direct}

Firstly consider the temporal projection $G^{(n)t}{}_{t}$. The free index $t$ occupies one slot in each row of the generalised delta, and the delta vanishes unless all upper indices and all lower indices are separately distinct; hence $t$ is excluded from every contracted position. By
\eqref{eq:blocks} the only blocks available are the sphere block and the radial block $R^{ri}{}_{rj}$. Moreover the radial direction, like $t$, can occupy at most one slot in each row, so no term containing two radial blocks survives. This leaves exactly two families. 

\begin{description}
\item{(i)} The angular family: Every one of the $n$ curvature factors is a
sphere block. Substituting $R^{ij}{}_{kl}=\psi\,\delta^{ij}_{kl}$ into
\eqref{eq:lovelocktensor}, each factor contributes $\psi$ and a factor $2$
from the antisymmetry of $\delta^{ij}_{kl}$; the $n$ factors of $2$ cancel
$2^{-n}$, leaving the overall $\tfrac12$. The remaining contraction is
that of a rank-$2n$ generalised delta over the angular directions, so by
\eqref{eq:deltacontraction}
\begin{equation}
  -2G^{(n)t}{}_{t}\Big|_{\text{ang}}
  =\ff{d-2}{2n}\,\psi^{n}
  \;=\;A_{n}\,\psi^{n},
  \qquad A_{n}\equiv\ff{d-2}{2n}.
\label{eq:angfamily}
\end{equation}
The count is transparent: $2n$ distinct angular directions must be chosen,
in order, from the $d-2$ that exist.

\item{(ii)} The single radial family: Exactly one factor is
$R^{ri}{}_{rj}=-\tfrac12 b\,\delta^{i}{}_{j}$, which may be assigned to
any of the $n$ curvature positions; the remaining $n-1$ factors are sphere
blocks. The radial block supplies $-\tfrac12 b$ and uses up one angular
direction rather than two, so the contraction now runs over $2n-1$ angular
directions and
\begin{equation}
  -2G^{(n)t}{}_{t}\Big|_{\text{rad}}
  =-\,n\,\ff{d-2}{2n-1}\,b\,\psi^{n-1}
  \;=\;-\,B_{n}\,b\,\psi^{n-1},
  \qquad B_{n}\equiv n\,\ff{d-2}{2n-1}.
\label{eq:radfamily}
\end{equation}
Adding the two families gives the following.
\end{description}

\begin{proposition}[Temporal projection]
\label{prop:weights}
For every $n\ge1$,
\begin{equation}
  -2G^{(n)t}{}_{t}
  =K_{n}(d)\Bigl[(d-2n-1)\psi^{n}-n\,b\,\psi^{n-1}\Bigr],
  \qquad
  K_{n}(d)\equiv\ff{d-2}{2n-1},
\label{eq:Gttn}
\end{equation}
since $A_{n}=\ff{d-2}{2n}=(d-2n-1)\,\ff{d-2}{2n-1}=(d-2n-1)K_{n}(d)$ and
$B_{n}=nK_{n}(d)$.
\end{proposition}

It is worth noting that it is essential to retain the falling-factorial form. The equivalent expression $A_{n}=(d-2)!/(d-2n-2)!$ is formally undefined in the critical
odd dimension $d=2n+1$, where it involves $(-1)!$, whereas $\ff{d-2}{2n}$ remains meaningful there and vanishes manifestly through its last factor
$d-2n-1$.

\subsection{Inductive argument}
\label{sec:induction}

The coefficients $A_{n}$ and $B_{n}$ obtained above may be confirmed
independently, without re-using \eqref{eq:deltacontraction}, by
establishing recurrences in the curvature order.

\begin{proposition}[Recurrences]
\label{prop:induction}
The coefficients defined in \eqref{eq:angfamily} and \eqref{eq:radfamily}
satisfy
\begin{equation}
  A_{n+1}=(d-2n-2)(d-2n-3)\,A_{n},
  \qquad
  B_{n+1}=\frac{n+1}{n}\,(d-2n-1)(d-2n-2)\,B_{n},
\label{eq:recurrences}
\end{equation}
with base values $A_{1}=(d-2)(d-3)$ and $B_{1}=d-2$. Consequently
$A_{n}=\ff{d-2}{2n}$ and $B_{n}=n\ff{d-2}{2n-1}$ for all $n\ge1$.
\end{proposition}

\begin{proof}
For the base case, $n=1$ gives
$-2G^{(1)t}{}_{t}=(d-2)\bigl[(d-3)\psi-b\bigr]$, which is the Einstein
component in the present notation; hence $A_{1}=(d-2)(d-3)$ and
$B_{1}=d-2$.

For the induction step, we move  from order $n$ to order $n+1$ by adjoining one further curvature factor. In the all-angular family the additional factor is a sphere block, which requires two further distinct angular directions and therefore multiplies the count by $(d-2n-2)(d-2n-3)$, giving the first recurrence. In the one-radial family the additional factor is likewise a sphere block, using two further directions and multiplying the count
by $(d-2n-1)(d-2n-2)$ --- one factor larger than in the all-angular case, because the radial block has already used up  only one direction --- while the number of positions available to the radial block increases from $n$
to $n+1$, contributing the ratio $(n+1)/n$. This is the second recurrence. Both recurrences are solved by the stated falling factorials, as is verified directly by
$\ff{m}{p+2}=(m-p)(m-p-1)\ff{m}{p}$ with $m=d-2$.
\end{proof}

Propositions~\ref{prop:weights} and \ref{prop:induction} are independent
routes to the same coefficients: the first contracts the delta once and
for all, the second propagates the count order by order.

\subsection{The radial projection}
\label{sec:radial}

The radial equation follows from the same argument with the free index
changed, and does not require a separate calculation.

\begin{corollary}[Radial projection]
\label{cor:radial}
For every $n\ge1$,
\begin{equation}
  -2G^{(n)r}{}_{r}
  =K_{n}(d)\Bigl[(d-2n-1)\psi^{n}-n\,a\,\psi^{n-1}\Bigr].
\label{eq:Grrn}
\end{equation}
\end{corollary}

\begin{proof}
With $r$ the free index, the radial direction is excluded from every
contracted position, so the radial block $R^{ri}{}_{rj}$ cannot appear;
neither can $R^{tr}{}_{tr}$, which carries a radial index. The surviving
families are therefore the all-angular family, unchanged, and a family
containing exactly one temporal--angular block
$R^{ti}{}_{tj}=-\tfrac12 a\,\delta^{i}{}_{j}$ with the remaining $n-1$
factors angular. The latter has the same index structure as the one-radial
family of Section~\ref{sec:direct}, with $t$ in place of $r$ and $a$ in
place of $b$, and therefore the same combinatorial weight $B_{n}$.
\end{proof}

Thus the conjugacy between the density and the radial pressure --- the
interchange $b\leftrightarrow a$ --- is a consequence of the block
structure \eqref{eq:blocks}, in which the temporal and radial blocks enter
symmetrically, rather than an asserted substitution. Note that the
temporal--radial block $R^{tr}{}_{tr}$, and hence $L$, contributes to
neither projection; it enters only the angular equations.

It is convenient to absorb the common combinatorial factor \(K_n(d)\) into the Lovelock couplings and define the dimensionally weighted couplings
\begin{equation}
  \widetilde\alpha_{n}=\alpha_{n}\,\ff{d-2}{2n-1}
  =\alpha_{n}(d-2)(d-3)\cdots(d-2n),
  \qquad \widetilde\alpha_{0}=\frac{\alpha_{0}}{d-1}.
\label{eq:weights}
\end{equation}
The weight contains the factor $d-2n$ and vanishes identically for
$d\le2n$: this is the truncation of the Lovelock series at
\begin{equation}
  N=\Bigl\lfloor\tfrac{d-1}{2}\Bigr\rfloor,
\label{eq:Nmax}
\end{equation}
obtained here by counting, since one cannot antisymmetrise over more
directions than the sphere possesses.

It is convenient to fix terminology at this point. By a \emph{Lovelock
theory} in $d$ dimensions we shall mean a choice of the coupling vector
\begin{equation}
  \bigl(\alpha_{0},\alpha_{1},\ldots,\alpha_{N}\bigr),
\label{eq:couplingvector}
\end{equation}
no further couplings being available since $\widetilde\alpha_{n}=0$ for
$n>N$. The theories usually named separately are particular coupling
vectors rather than distinct frameworks: Einstein gravity is
$\alpha_{1}$ alone, Einstein--Gauss--Bonnet is $\alpha_{1}$ together with
$\alpha_{2}$, pure Lovelock of order $n$ is $\alpha_{n}$ alone, and what
is often called the full or complete theory has all of
$\alpha_{1},\ldots,\alpha_{N}$ non-zero. We write $n_{\max}$ for the
highest order actually present, $\widetilde\alpha_{n_{\max}}\neq0$, and
call the theory \emph{maximal} when $n_{\max}=N$. Statements asserted
below for every Lovelock theory are to be read as holding for every such
coupling vector; where a hypothesis on $n_{\max}$ is required it is
stated explicitly.

\section{Density and radial pressure}
\label{sec:rhopr}

Multiplying \eqref{eq:Gttn} and \eqref{eq:Grrn} by $\alpha_{n}$ and using
$\psi=C(1-Z)/x$, $b=2C\dot Z$, $a=4CZ\dot y/y$ gives the contribution of
the individual $n$-th Lovelock density,
\begin{equation}
  2\rho_{n}=\widetilde\alpha_{n}C^{n}
   \frac{(1-Z)^{n-1}}{x^{n}}
   \Bigl[(d-2n-1)(1-Z)-2nx\dot Z\Bigr],
\label{eq:rhon}
\end{equation}
\begin{equation}
  2p_{r,n}=\widetilde\alpha_{n}C^{n}
   \frac{(1-Z)^{n-1}}{x^{n}}
   \Bigl[4nxZ\frac{\dot y}{y}-(d-2n-1)(1-Z)\Bigr].
\label{eq:prn}
\end{equation}
Since the Lovelock tensor is linear in the independent couplings, the
cumulative equations are the sums of these contributions together with the
cosmological terms $2\rho_{0}=\alpha_{0}$, $2p_{r,0}=-\alpha_{0}$:
\begin{align}
  2\rho={}&\alpha_{0}+\sum_{n=1}^{N}\widetilde\alpha_{n}C^{n}
    \frac{(1-Z)^{n-1}}{x^{n}}
    \Bigl[(d-2n-1)(1-Z)-2nx\dot Z\Bigr],
\label{eq:rhosum}\\
  2p_{r}={}&-\alpha_{0}+\sum_{n=1}^{N}\widetilde\alpha_{n}C^{n}
    \frac{(1-Z)^{n-1}}{x^{n}}
    \Bigl[4nxZ\frac{\dot y}{y}-(d-2n-1)(1-Z)\Bigr].
\label{eq:prsum}
\end{align}

\section{Tangential pressure from conservation}
\label{sec:pt}

The angular projection need not be computed explicitly. By the contracted Bianchi
identity we obtain the continuity equation in the form
\begin{equation}
  p_{r}'+\frac{\nu'}{2}(\rho+p_{r})+\frac{d-2}{r}(p_{r}-p_{t})=0,
  \end{equation}
  which may be given equivalently as 
  \begin{equation}
  p_{t}=p_{r}+\frac{2x}{d-2}
   \Bigl[\dot p_{r}+\frac{\dot y}{y}(\rho+p_{r})\Bigr]
\label{eq:conservation}
\end{equation}
in Buchdahl coordinates. 
Two identities govern the reduction. The first, obtained by comparing the
two expressions for $\rho$ below, is
\begin{equation}
  r\psi'=-\bigl(b+2\psi\bigr),
\label{eq:rpsiprime}
\end{equation}
and the second follows from \eqref{eq:psiab} and \eqref{eq:L} by direct
differentiation,
\begin{equation}
  r a'=L-a-\frac{r^{2}}{2Z}\,a\,(a-b).
\label{eq:raprime}
\end{equation}

The mechanism which keeps the system tractable is a cancellation between
the two terms of \eqref{eq:conservation}. Writing
$2p_{r}=aW_{\psi}-V$ in the notation of Section~\ref{sec:resum} and
differentiating yields the relationship 
\begin{eqnarray}
  2r\,p_{r}'&=& (ra')W_{\psi}+(r\psi')\bigl(aW_{\psi\psi}-V_{\psi}\bigr) \nonumber \\ 
  &=&\Bigl[L-a-\frac{r^{2}}{2Z}a(a-b)\Bigr]W_{\psi}
   -\bigl(b+2\psi\bigr)\bigl(aW_{\psi\psi}-V_{\psi}\bigr),
\label{eq:rprprime}
\end{eqnarray}
while, using $\rho+p_{r}=\tfrac12(a-b)W_{\psi}$ and $\nu'=ar/Z$ gives the compact form 
\begin{equation}
  \frac{r\nu'}{2}\bigl(\rho+p_{r}\bigr)
  =\frac{r^{2}}{4Z}\,a\,(a-b)\,W_{\psi}
\label{eq:rnuprime}
\end{equation}
from which the inertial mass density $\rho + p_r$ may be found. The terms quadratic in $a$ in \eqref{eq:rprprime} and \eqref{eq:rnuprime}
are equal and opposite and cancel identically. Since $a$ is the only
quantity carrying $\dot y/y$, this cancellation is the reason that no
$(\dot y/y)^{2}$ term survives, that the angular equation remains linear
in $y$, and --- because $\ddot Z$ appears nowhere in \eqref{eq:rprprime}
--- that the system stays of first order in $Z$. What remains is
\begin{equation}
  2(d-2)\bigl(p_{t}-p_{r}\bigr)
  =\bigl(L-a\bigr)W_{\psi}+\bigl(b+2\psi\bigr)V_{\psi}
   -a\bigl(b+2\psi\bigr)W_{\psi\psi},
\label{eq:anisoraw}
\end{equation}
which is rearranged in Section~\ref{sec:Ind}. Carrying the same
substitution through order by order, the contribution of the $n$-th
density with $q_{n}=d-2n-1$ is
\begin{align}
  2(d-2)p_{t,n}={}&\widetilde\alpha_{n}C^{n}\Bigg\{
   4n\frac{(1-Z)^{n-1}}{x^{n-1}}
   \Bigl(2xZ\frac{\ddot y}{y}+Z\frac{\dot y}{y}+x\dot Z\frac{\dot y}{y}\Bigr)
   +2nq_{n}\frac{(1-Z)^{n-1}}{x^{n-1}}
     \Bigl(2Z\frac{\dot y}{y}+\dot Z\Bigr)
   \notag\\
   &\qquad\qquad
   -8n(n-1)\frac{(1-Z)^{n-2}}{x^{n-2}}Z\dot Z\frac{\dot y}{y}
   -q_{n}(q_{n}-1)\frac{(1-Z)^{n}}{x^{n}}\Bigg\}.
\label{eq:ptn}
\end{align}
\new{%
For a single active order \eqref{eq:ptn} is the transverse pressure
obtained by the same conservation argument in
\cite{DadhichHansrajMaharaj2016}; what follows is the summation of the
independently coupled hierarchy built on it.}

\section{Performing the summation: three polynomials}
\label{sec:resum}

The sums \eqref{eq:rhosum}, \eqref{eq:prsum} and \eqref{eq:ptn} are
generated by the characteristic polynomial together with two polynomials
obtained from it by a single differential operation,
\begin{equation}
  W(\psi)=\sum_{n=0}^{N}\widetilde\alpha_{n}\psi^{n},
  \qquad
  V=(d-1)W-2\psi W_{\psi},
  \qquad
  U=(d-2)V-2\psi V_{\psi},
\label{eq:WVU}
\end{equation}
so that, expanding,
\begin{equation}
  V=\sum_{n}\widetilde\alpha_{n}(d-1-2n)\psi^{n},
  \qquad
  U=\sum_{n}\widetilde\alpha_{n}(d-1-2n)(d-2-2n)\psi^{n},
  \qquad
  V_{\psi}=(d-3)W_{\psi}-2\psi W_{\psi\psi}.
\label{eq:VUexp}
\end{equation}
Each application of $P\mapsto kP-2\psi P_{\psi}$ multiplies the $n$-th
coefficient by $k-2n$ and lowers the label $k$ by one.

\new{%
The first member of the sequence is not new. The combination $V$ is
proportional to the effective mean density of the quasi-topological
literature, $(d-3)\bar\varrho^{\,\rm eff}=V/2$ in the shared normalisation
$\alpha_{0}=0$, which controls the vacuum branch structure there; the
source in the equation for the lapse is governed by the product
$\dot\psi V_{\psi}$ rather than by $V_{\psi}$ alone \cite{Bueno2026}. What is distinctive here is that the sequence does not
stop at $V$: the tangential pressure requires the third member $U$, and
the assertion of Theorem~\ref{thm:closed} is that the entire matter
tensor, and not only its isotropic part, closes on $W\mapsto V\mapsto U$
together with the first two derivatives of $W$.}

\begin{theorem}[Closed form]
\label{thm:closed}
The field equations \eqref{eq:rhosum}, \eqref{eq:prsum} and the sum of
\eqref{eq:ptn} are equivalent to
\begin{equation}
  2\rho=V-bW_{\psi},
  \qquad
  2p_{r}=aW_{\psi}-V,
  \qquad
  2(d-2)p_{t}=LW_{\psi}+(a+b)V_{\psi}-abW_{\psi\psi}-U .
\label{eq:closed}
\end{equation}
The density is moreover an exact total derivative,
$2\rho\,r^{d-2}=(r^{d-1}W)'$, whence
\begin{equation}
  m(r)=\tfrac12 r^{d-1}W(\psi),\qquad \rho=\frac{m'}{r^{d-2}} .
\label{eq:mass}
\end{equation}
\end{theorem}

Comparing the two expressions for $\rho$ yields \eqref{eq:rpsiprime}, and
subtracting the first two of \eqref{eq:closed} gives
\begin{equation}
  \rho+p_{r}=\frac{a-b}{2}\,W_{\psi}
   =C\Bigl(2Z\frac{\dot y}{y}-\dot Z\Bigr)W_{\psi},
\label{eq:inertial}
\end{equation}
the inertial mass density of the full Lovelock fluid. The coefficient
relating $\rho+p_{r}$ to the metric combination $a-b$ is $W_{\psi}/2$, and
it is this coefficient which degenerates on a branch with
$W_{\psi}=0$; note that $\rho+p_{r}$ also vanishes on the ordinary
configurations with $a=b$, that is $Z\nu'=Z'$, so the vanishing of the
inertial density does not by itself signal a degenerate branch.

Eliminating $V$ and $U$ from the third of \eqref{eq:closed}, and writing
the result in the variables $(x,Z,y)$ with
\begin{align}
  \mathcal A={}&8xZ\frac{\ddot y}{y}+4(d-2)Z\frac{\dot y}{y}
   +4x\dot Z\frac{\dot y}{y}+2(d-3)\dot Z+\frac{2(2d-5)(1-Z)}{x},
\label{eq:Acal}\\
  \mathcal B={}&8Z\dot Z\frac{\dot y}{y}
   +\frac{8(1-Z)Z}{x}\frac{\dot y}{y}
   +\frac{4(1-Z)\dot Z}{x}+\frac{4(1-Z)^{2}}{x^{2}},
\label{eq:Bcal}
\end{align}
one obtains the compact statement
\begin{equation}
  2(d-2)p_{t}=C\mathcal A\,W_{\psi}-C^{2}\mathcal B\,W_{\psi\psi}
   -(d-1)(d-2)W,\qquad \psi=C\frac{1-Z}{x}.
\label{eq:ptcompact}
\end{equation}
Equations \eqref{eq:closed} close the complete static matter sector for an
arbitrary Lovelock polynomial. The convexity $W_{\psi\psi}$ enters exactly
once and is absent in Einstein gravity, where $W$ is linear.

It helps to note that since $p_r$ is linear in $\nu'$ with slope $ZW_{\psi}/2r$, evaluating an interior and a vacuum exterior at the same $(R,Z,\nu)$ and subtracting
gives $p_{r}(R^{-})=\bigl[Z(R)W_{\psi}(\psi_{R})/2R\bigr][\nu']_{\Sigma}$ whenever $[\nu]_{\Sigma}=[Z]_{\Sigma}=0$. On non-degenerate branches the vanishing of the surface pressure and the continuity of $\nu'$ are therefore the same statement. This is an identity of the field equations
and not a junction analysis; the surface tensor of Lovelock gravity \cite{Davis2003,Gravanis2003} is nonlinear in the jump of the extrinsic curvature, and the counting of independent matching conditions is not
treated here.

\section{Anisotropy and the pure-order isotropy operator}
\label{sec:Ind}

Rearranging \eqref{eq:anisoraw} with the aid of \eqref{eq:VUexp} gives the
pressure anisotropy $\Delta=p_{t}-p_{r}$ in the form
\begin{equation}
  2(d-2)\Delta=\bigl[L-a+(d-3)(b+2\psi)\bigr]W_{\psi}
   -(a+2\psi)(b+2\psi)W_{\psi\psi},
\label{eq:aniso}
\end{equation}
the entire Lovelock polynomial entering through $W_{\psi}$ and
$W_{\psi\psi}$ only, with the cosmological term dropping out identically.
The isotropy condition $\Delta=0$ also possesses a total-derivative form,
\begin{equation}
  \frac{d}{dr}\Bigl(W_{\psi}\frac{\sqrt Z\,y'}{r}\Bigr)
   =\frac{\psi' V_{\psi}}{2\sqrt Z}\,y
  \qquad\Longleftrightarrow\qquad
  \frac{d}{dx}\bigl(W_{\psi}\sqrt Z\,\dot y\bigr)
   +\frac{V_{\psi}}{4x\sqrt Z}\Bigl(\dot Z+\frac{1-Z}{x}\Bigr)y=0,
\label{eq:totalderiv}
\end{equation}
using $\dot\psi=-(C/x)\bigl(\dot Z+(1-Z)/x\bigr)$.

Equation \eqref{eq:totalderiv} is not itself new. For a perfect fluid the
same equation is obtained in \cite{Bueno2026} from the two field
equations together with the conservation law, in the form
$\bigl(W_{\psi}\zeta_{,\xi}\bigr)_{,\xi}=g\zeta$ with
$d\xi=dx/\sqrt Z$, $\zeta=y$ and $4g=\dot\psi V_{\psi}$; it is the
higher-curvature descendant of the equation Buchdahl used in general
relativity. What the present derivation adds is to embed that equation in
the full anisotropic system: it is obtained by equating the second and third equations in 
\eqref{eq:closed} so that it is
identified as the locus $\Delta=0$ of the anisotropy \eqref{eq:aniso}
rather than as a consequence of assuming a perfect fluid; the anisotropy itself, which is what an anisotropic source
requires and which the redundancy of the angular equation under the
perfect-fluid assumption conceals, is then available in the same closed
form.

The central result is that the anisotropy decomposes over the Lovelock
orders into copies of a single operator whose dependence on the theory is
through two integers only.

\begin{theorem}[Pure-order isotropy operator]
\label{thm:Ind}
Define
\begin{equation}
  \begin{aligned}
    \mathcal I_{n,d}[Z,y]={}&4x^{2}Z(1-Z)\,\ddot y\\
    &+2x\Bigl\{x\bigl[1-(2n-1)Z\bigr]\dot Z-2(n-1)Z(1-Z)\Bigr\}\dot y\\
    &+(d-2n-1)(1-Z)\bigl[x\dot Z+1-Z\bigr]y .
  \end{aligned}
\label{eq:Ind}
\end{equation}
Then for an arbitrary Lovelock polynomial
\begin{equation}
  (d-2)\,\Delta=\frac{1}{y}\sum_{n=1}^{N}
   n\,\widetilde\alpha_{n}C^{n}\frac{(1-Z)^{n-2}}{x^{n}}\,
   \mathcal I_{n,d}[Z,y],
\label{eq:anisoInd}
\end{equation}
so that the isotropy condition of the full theory reads
\begin{equation}
  \sum_{n=1}^{N}n\,\widetilde\alpha_{n}C^{n}
   \frac{(1-Z)^{n-2}}{x^{n}}\,\mathcal I_{n,d}[Z,y]=0 .
\label{eq:isoInd}
\end{equation}
\end{theorem}

Equation \eqref{eq:isoInd} is equivalent to \eqref{eq:aniso} and
\eqref{eq:totalderiv}, but displays the contribution of each Lovelock
order separately while \eqref{eq:Ind} supplies the universal building
block. Three features are worth noting. The operator
$\mathcal I_{1,d}$ carries an overall factor $(1-Z)$, which cancels the
$(1-Z)^{-1}$ in the $n=1$ weight of \eqref{eq:isoInd}, so no negative
powers of $1-Z$ occur. The dimension enters $\mathcal I_{n,d}$ only
through the coefficient $d-2n-1$ of the undifferentiated $y$; the
second- and first-derivative coefficients are dimension-independent.
And the order $n$ enters only through the two integers $2n-1$ and $n-1$ in
the $\dot y$ coefficient, together with $d-2n-1$.

\new{%
It is useful to record \eqref{eq:isoInd} in the form in which a
pressure-isotropy equation is ordinarily displayed, as a linear
second-order equation for the temporal potential. Collecting the
derivatives of $y$ in \eqref{eq:isoInd} and writing the result as
$\mathsf P\ddot y+\mathsf Q\dot y+\mathsf Sy=0$, one finds that the whole
of the theory dependence is carried by just three quantities, $W_{\psi}$,
$W_{\psi\psi}$ and $V_{\psi}$, with $\mathsf P$ controlled by the first
alone and $\mathsf S$ by the third alone, while $\mathsf Q$ mixes the
first two,
\begin{equation}
  \mathsf P=4CxZ\,W_{\psi},
  \qquad
  \mathsf Q=2Cx\dot Z\,W_{\psi}-2CZ(b+2\psi)\,W_{\psi\psi},
  \qquad
  \mathsf S=-x\dot\psi\,V_{\psi},
\label{eq:PQS}
\end{equation}
the combination $b+2\psi$ of \eqref{eq:aniso} reappearing in $\mathsf Q$.
Equation \eqref{eq:PQS} is \eqref{eq:totalderiv} expanded, and it shows
that the degeneracies of the isotropy condition are governed by the
characteristic polynomial and not by the choice of potential. The
vanishing of $\mathsf S$ separates into $\dot\psi=0$, treated in
Section~\ref{sec:uniform}, and $V_{\psi}=0$, treated in
Section~\ref{sec:nogo}. The vanishing of $\mathsf P$ is algebraic rather
than differential. Should $\mathsf P$ vanish identically on an interval,
then since $W_{\psi}$ is a non-trivial polynomial its zeros are discrete,
so $\psi$ is confined to one of them and is constant there; whereupon
$b+2\psi=0$ and $\dot\psi=0$ by \eqref{eq:rpsiprime}, and $\mathsf Q$ and
$\mathsf S$ vanish with it. At an isolated radius no such conclusion
follows, the solution merely passing through a zero of $W_{\psi}$. On an
interval where $W_{\psi}(\psi)=0$, that is where the relevant root of the
Wheeler equation $W(\psi)=2m/r^{d-1}$ is a multiple root, the isotropy
condition therefore degenerates to $0=0$ and ceases to determine the
temporal potential at all. This is the precise sense in which the
non-degeneracy hypothesis $W_{\psi}\neq0$ is needed throughout what
follows.
}

\section{Checks in five, six and seven dimensions}
\label{sec:checks}

We now undertake some verifications with known results in the literature. 
Take $\alpha_{1}=1$, $\alpha_{2}=\alpha$ and define
\begin{align}
  \mathcal E_{d}={}&4x^{2}Z\ddot y+2x^{2}\dot Z\dot y
   +(d-3)(x\dot Z+1-Z)y,
\label{eq:Ed}\\
  \mathcal G_{d}={}&4x^{2}Z(1-Z)\ddot y
   +2x\bigl[x(1-3Z)\dot Z-2Z(1-Z)\bigr]\dot y
   +(d-5)(1-Z)(x\dot Z+1-Z)y,
\label{eq:Gd}
\end{align}
which are $\mathcal I_{1,d}/(1-Z)$ and $\mathcal I_{2,d}$ respectively.
Equation \eqref{eq:isoInd} then reduces to
\begin{equation}
  x\,\mathcal E_{d}+2\alpha C(d-3)(d-4)\,\mathcal G_{d}=0 .
\label{eq:EGBd}
\end{equation}
For $d=5$, after division by $2x$,
\begin{equation}
  0=2x^{2}Z\ddot y+x^{2}\dot Z\dot y+(x\dot Z-Z+1)y
   +4\alpha C\Bigl[2xZ(1-Z)\ddot y
   +\bigl\{x(1-3Z)\dot Z-2Z(1-Z)\bigr\}\dot y\Bigr],
\label{eq:EGB5}
\end{equation}
the undifferentiated term of $\mathcal G_{5}$ being absent because
$d-5=0$; and for $d=6$,
\begin{align}
  0={}&x\Bigl[4x^{2}Z\ddot y+2x^{2}\dot Z\dot y+3(x\dot Z-Z+1)y\Bigr]
  \notag\\
  &+12\alpha C\Bigl[4x^{2}Z(1-Z)\ddot y
   +2x\bigl\{x(1-3Z)\dot Z-2Z(1-Z)\bigr\}\dot y
   +(1-Z)(x\dot Z-Z+1)y\Bigr].
\label{eq:EGB6}
\end{align}
These are the standard five- and six-dimensional Einstein--Gauss--Bonnet
pressure-isotropy equations in the variables $x=Cr^{2}$, $Z=e^{-\lambda}$,
$e^{\nu}=y^{2}$. Observe that the coefficients of $\ddot y$ and $\dot y$
agree in \eqref{eq:EGB5} and \eqref{eq:EGB6} up to the overall factor $x$;
only the coefficient of $y$ differs, through $d-2n-1$, exactly as
Theorem~\ref{thm:Ind} requires.

The cubic case is instructive. At $d=7$ and $n=3$ the coefficient
$d-2n-1$ vanishes, so by \eqref{eq:rhon} and \eqref{eq:prn} the
third-order density and radial pressure lose their algebraic
$\psi^{3}$ terms entirely and survive only through the derivative terms:
\begin{equation}
  2\rho_{3}\big|_{d=7}=-6\,\widetilde\alpha_{3}C^{3}
    \frac{(1-Z)^{2}\dot Z}{x^{2}},
  \qquad
  2p_{r,3}\big|_{d=7}=12\,\widetilde\alpha_{3}C^{3}
    \frac{(1-Z)^{2}Z}{x^{2}}\frac{\dot y}{y},
\label{eq:cubic7}
\end{equation}
while $\mathcal I_{3,7}$ has no undifferentiated-$y$ term,
\begin{equation}
  \mathcal I_{3,7}=4x^{2}Z(1-Z)\ddot y
   +2x\bigl[x(1-5Z)\dot Z-4Z(1-Z)\bigr]\dot y .
\label{eq:I37}
\end{equation}
Equations \eqref{eq:cubic7} and \eqref{eq:I37} are recorded here to
facilitate comparison with published third-order Lovelock stellar models.

\subsection*{The third-order equations in arbitrary dimension}

Such a comparison is now available in full. The third-order Lovelock field
equations for a static sphere have been obtained in arbitrary dimension by
direct computation \cite{NaickerBrasselMaharaj2026}, in the presence of an
electromagnetic field. Setting $C=1$, so that $x=r^{2}$, the couplings
used there are $\hat\alpha_{2}=\alpha_{2}(d-3)(d-4)$ and
$\hat\alpha_{3}=\alpha_{3}(d-3)(d-4)(d-5)(d-6)$, which are
$\widetilde\alpha_{n}=(d-2)\hat\alpha_{n}$ in the present normalisation,
in agreement with the falling factorials of \eqref{eq:weights}.
Multiplying \eqref{eq:isoInd}, truncated at $N=3$, by $x^{3}/(d-2)$
returns their pressure-isotropy equation identically, the coefficient
$d-2n-1$ of the undifferentiated potential appearing there as $d-3$,
$d-5$ and $d-7$ for the three orders; their density and radial pressure
follow in the same way from \eqref{eq:rhon} and \eqref{eq:prn}. The
electromagnetic field enters as an additive anisotropic source and leaves
the operator \eqref{eq:Ind} untouched, so that the neutral content of that
system is the $N=3$ specialisation of \eqref{eq:isoInd}.

The comparison also accounts for a device employed there. Exact solutions are generated in  \cite{NaickerBrasselMaharaj2026} by setting each of the
three coefficients of the isotropy equation to zero in turn, a procedure the authors describe as ad hoc. That strategy is analogous in spirit to Tolman's classical use of simplifying ansätze to generate exact perfect-fluid solutions \cite{tolman1939}. Equation~\eqref{eq:PQS} shows that, in the present variables, the three choices correspond to distinct degeneracies of the
characteristic-polynomial coefficients.  By \eqref{eq:PQS} the three coefficients are governed by $W_{\psi}$, $W_{\psi\psi}$, $V_{\psi}$ and the radial
variation of $\psi$, so the trichotomy is not a choice of ansatz but a
reading of the degeneracy structure of the characteristic polynomial. The
last coefficient vanishes for constant $\psi$, as in the uniform density
interior of Section~\ref{sec:uniform}, or wherever $V_{\psi}=0$, of which
the critical pure orders of Section~\ref{sec:critical} furnish the
identically degenerate example; the first vanishes on the degenerate
algebraic branch where the relevant root of the Wheeler equation is
multiple.
That this is visible only after the summation has been performed is
perhaps the clearest illustration of what the closed form is for.

The six-dimensional density, radial-pressure and isotropy equations also agree identically with an independent Maple/GRTensor calculation.  A symbolic residual comparison with the published third-order system,
leaving $d$ and the couplings free, confirms the corresponding cubic equations.

\section{Uniform density and the universal interior}
\label{sec:uniform}

Let $\rho=\rho_{0}$ be constant. Integrating the total-derivative form
\eqref{eq:mass},
\begin{equation}
  W(\psi)=\frac{2\rho_{0}}{d-1}+\frac{M_{0}}{r^{d-1}},
\label{eq:constW}
\end{equation}
with $M_{0}$ an integration constant. Regularity at the centre requires
$M_{0}=0$, whence $W(\psi)$ is constant throughout the fluid and therefore
\begin{equation}
  W_{\psi}\,\psi'=0 .
\label{eq:Wpsipsiprime}
\end{equation}
\new{Since $W$ is a non-constant polynomial, $W(\psi)-W(\psi_{0})$ has
finitely many roots, and a continuous $\psi$ taking values in a finite set
on a connected interior is constant; the conclusion therefore does not
require $W_{\psi}\neq0$, and holds even at a multiple root.} This gives
$\psi'=0$, equivalently $b+2\psi=0$ by \eqref{eq:rpsiprime}, and hence
\begin{equation}
  \psi=\psi_{0}=\text{const},
  \qquad Z=1-\psi_{0}r^{2},
\label{eq:uniformZ}
\end{equation}
irrespective of dimension and couplings. This is the universality
established by Dadhich, Molina and Khugaev
\cite{DadhichMolinaKhugaev2010}, who proved it for Einstein and for
Einstein--Gauss--Bonnet gravity and argued that it should extend to
Lovelock gravity generally. That expectation is correct, and the reason is
visible in \eqref{eq:isoInd}: the isotropy condition is a linear
combination of the operators $\mathcal I_{n,d}$ with the weighted
couplings as coefficients, so a metric annihilating each operator
separately annihilates every combination of them. Universality at each
curvature order therefore propagates to an arbitrary Lovelock polynomial
without further hypothesis. \new{Regularity is
essential, since without it $\psi$ need not be constant; non-degeneracy is
not needed here, and enters only at the next step, where $W_{\psi}\neq0$
is what allows the isotropy condition to determine the temporal potential
uniquely.}

With $\psi'=0$ the right-hand side of \eqref{eq:totalderiv} vanishes, so
$W_{\psi}\sqrt Z\,y'/r$ is constant and a single integration gives the
general solution.

\begin{theorem}[Universal interior]
\label{thm:universal}
For $\psi_{0}\neq0$, and on a non-degenerate branch $W_{\psi}(\psi_{0})
\neq0$, the general regular solution of the isotropy condition at uniform
density is the interior Schwarzschild metric
\begin{equation}
  e^{-\lambda}=1-\psi_{0}r^{2},
  \qquad
  e^{\nu/2}=y=A+B\sqrt{1-\psi_{0}r^{2}},
\label{eq:universalstar}
\end{equation}
and it solves that condition in every Lovelock theory, in every dimension
and for every set of couplings simultaneously.
\end{theorem}

\new{%
The metric \eqref{eq:universalstar} remains a solution when
$W_{\psi}(\psi_{0})=0$, since by \eqref{eq:PQS} all three coefficients of
the isotropy condition then vanish; what is lost on such a branch is
uniqueness, the condition being satisfied by an arbitrary temporal
potential.}

Both potentials, and not merely the spatial one, are thus common to the
whole family; the reason is visible in \eqref{eq:totalderiv}, uniform
density annihilating the only channel through which the couplings act.

Since $W$ is a polynomial of degree $N$, \eqref{eq:constW} fixes
$\psi_{0}$ algebraically and not by integration: at uniform density the
determination of the spatial potential reduces entirely to a root
problem, no differential equation for $\psi$ remaining, while the
isotropy condition continues to determine the temporal potential.

Which root to take is a question of continuation in the couplings. As the
higher couplings are switched off, \eqref{eq:constW} degenerates to the
linear equation $\widetilde\alpha_{1}\psi_{0}=2\rho_{0}/(d-1)$, whose
solution is $\psi_{0}=2\rho_{0}/[(d-1)(d-2)]$ in the normalisation
$\alpha_{1}=1$, $\alpha_{0}=0$ used throughout. A root which tends to
this value along such a path may be called Einstein-connected, and it is
the branch selected here whenever one exists. Root coalescences along the
path are located by the discriminant of $W(\psi)-2\rho_{0}/(d-1)$, which
for $N\geq4$ does not by itself fix how many roots are real; the other
roots are not excluded by any general obstruction, and we do not claim
them to be inadmissible. Once Theorem~\ref{thm:universal} is available, a treatment of the
constant density interior at a fixed curvature order reduces to a root
count; the cubic instance has been worked out explicitly in
\cite{MaharajNaickerBrassel2026}.

The matter variables at constant density are known in closed form for an
arbitrary characteristic function \cite{Bueno2026}, and the present
formalism reproduces them. Writing
\begin{equation}
  \sigma=\frac{2\psi_{0}W_{\psi}(\psi_{0})}{(d-1)W(\psi_{0})},
  \qquad\text{so that}\qquad
  1-\sigma=\frac{V(\psi_{0})}{(d-1)W(\psi_{0})},
\label{eq:sigma}
\end{equation}
one has $\rho_{0}=\tfrac12(d-1)W(\psi_{0})$ by \eqref{eq:mass}, and
evaluating \eqref{eq:closed} on \eqref{eq:universalstar} with a
pressure-free surface at $r=R$, which fixes
$B=A(\sigma-1)/\sqrt{Z_{R}}$ with $Z_{R}=Z(R)$, gives
\begin{equation}
  p_{r}=p_{t}=\rho_{0}\,
  \frac{(1-\sigma)\bigl(\sqrt{Z}-\sqrt{Z_{R}}\bigr)}
       {(\sigma-1)\sqrt{Z}+\sqrt{Z_{R}}},
  \qquad
  Z(R)>(1-\sigma)^{2},
\label{eq:uniformpr}
\end{equation}
the inequality being the condition for finite positive central pressure
on the stellar branch $\rho_{0}>0$, $\psi_{0}>0$, $0<\sigma<1$.
The quantity $\sigma$ is related to the ratio $\Delta_{\rm B}$ of
effective to actual mean density of \cite{Bueno2026}, written here with a
subscript to distinguish it from the anisotropy \eqref{eq:aniso}, by
$\Delta_{\rm B}=(d-1)(1-\sigma)/(d-3)$ in the shared normalisation
$\alpha_{0}=0$, equivalently
$\sigma=1-(d-3)\Delta_{\rm B}/(d-1)$, and in those circumstances  \eqref{eq:uniformpr} is their pressure profile
and their divergent-central-pressure limit; in general relativity
$W=(d-2)\psi$ gives $\sigma=2/(d-1)$ irrespective of the density, so the
inequality reads $Z(R)>\tfrac19$ at $d=4$, equivalently a compactness
below $\tfrac89$, which is Buchdahl's bound. We record the
correspondence because the recovery is a check on \eqref{eq:closed} in a
regime where an independent closed-form answer exists, and because
$\sigma$ is the natural variable in which to state what follows.

The zero-pressure endpoint itself is identified in \cite{Bueno2026}. What
the present formalism adds is its Lovelock-polynomial interpretation and
the classification of the corresponding identity. Equality $\sigma=1$,
which is the vanishing of the effective density and that pressureless
limit, holds exactly when
$V(\psi_{0})=0$. For $\psi_{0}>0$ and non-negative couplings with at
least one curvature coupling positive, the coefficients of $V$ are
non-negative and a positive root can occur only if $V$ vanishes
identically, which by Proposition~\ref{prop:Vflat} means pure Lovelock
gravity of order $N$ in the critical dimension $d=2N+1$; without those
hypotheses $V(\psi_{0})=0$ is a pointwise condition and
Proposition~\ref{prop:Vflat}, which classifies the identity, does not
apply. There $B=0$, the temporal potential is constant, and
the pressure vanishes identically. The pressureless limit and the known
absence of bounded pure Lovelock spheres in the critical odd dimension
are therefore the same phenomenon. The polynomial identity $V\equiv0$
cannot occur while the Einstein term is active, so this is not visible in
a normalisation that retains it; a pointwise root $V(\psi_{0})=0$ can
still occur with couplings of indefinite sign, away from the critical
dimension, and the degeneration occurs there instead.

At this point the following remark is in order. 
The case $\psi_{0}=0$ is exceptional. The Wronskian of the pair $\{1,\sqrt{1-\psi_{0}r^{2}}\}$ is proportional to $\psi_{0}$, so at
$\psi_{0}=0$ the two solutions coincide and \eqref{eq:universalstar}
degenerates to a constant. The second solution is recovered by
rescaling $B$ as $\psi_{0}\to0$, and the two-constant solution there is
$y=A+B'x$, with $Z\equiv1$. For $B'\neq0$ this is not flat spacetime: by
\eqref{eq:closed}, $\rho=\alpha_{0}/2$ while
$p_{r}=2CB'\widetilde\alpha_{1}/y-\alpha_{0}/2$, which is non-zero for
$B'\neq0$. It is instead the branch of vanishing orbit sectional curvature, and it holds irrespective of the theory for a second reason: with $Z\equiv1$
every weighted contribution to \eqref{eq:isoInd} of order $n\geq2$ vanishes identically, while the $n=1$ weight carries a factor $(1-Z)^{-1}$ whose singularity is removable against the overall factor $1-Z$ in $\mathcal I_{1,d}$. Rather than pass to that limit it is safer
to read the conclusion off \eqref{eq:PQS} directly, which gives $4CxW_{\psi}(0)\ddot y=0$. The conclusion 
$y=A+B'x$ therefore requires $W_{\psi}(0)=\widetilde\alpha_{1}\neq0$; if
the Einstein term is absent the condition degenerates completely and $y$
is arbitrary.

\section{Uniqueness of the constant-density interior}
\label{sec:uniqueness}

Theorem~\ref{thm:universal} asserts that one interior is insensitive to
the couplings. As noted there, this belongs to
\cite{DadhichMolinaKhugaev2010}, where it is stated for the general
Lovelock polynomial, where constancy of the density is shown to be
necessary as well as sufficient, and where it is already observed that
universality characterises the Schwarzschild interior of a sphere of
finite radius. The formulation developed here does not add that
conclusion; what it adds is a proof of it. The argument in
\cite{DadhichMolinaKhugaev2010} rests on the quasilinearity of the
Lovelock equations, whereas \eqref{eq:isoInd} permits the converse to be
settled directly, at fixed dimension and in coupling space, by asking
which geometries annihilate every $\mathcal I_{n,d}$ at once. Two further
things follow from that route which the earlier treatments do not show:
an exceptional branch, and the fact that two distinct curvature orders
already force the conclusion, so that neither the full hierarchy nor any
particular dimension is needed. Since a metric solves the
isotropy condition for arbitrary couplings precisely when it annihilates
every $\mathcal I_{n,d}$ separately, the question reduces to a
finite computation, because $\mathcal I_{n,d}$ is \emph{affine} in $n$.
Collecting powers of $n$ in \eqref{eq:Ind},
\begin{equation}
  \mathcal I_{n,d}=\mathcal I^{(0)}_{d}
   -2n\,\bigl(x\dot Z+1-Z\bigr)\bigl[2xZ\dot y+(1-Z)y\bigr],
\label{eq:Isplit}
\end{equation}
with $\mathcal I^{(0)}_{d}$ the $n$-independent remainder. Two distinct
orders therefore determine both parts, and we obtain

\begin{theorem}[Uniqueness]
\label{thm:uniqueness}
Let a static isotropic interior with $Z>0$ satisfy the pressure isotropy
condition for all values of the couplings of any two distinct Lovelock
orders. Then, apart from the branch $\psi\equiv0$ of vanishing orbit
sectional curvature described earlier it is the
interior Schwarzschild metric \eqref{eq:universalstar}. In particular the
interior Schwarzschild metric is the unique interior with $\psi\not\equiv0$
common to every Lovelock theory, and Einstein together with Gauss--Bonnet
already suffices to force it.
\end{theorem}

\begin{proof}
By \eqref{eq:Isplit} the vanishing of $\mathcal I_{n,d}$ for two distinct
orders is equivalent to the vanishing of the coefficient of $n$ and of
$\mathcal I^{(0)}_{d}$. The former factorises, so either
\begin{equation}
  x\dot Z+1-Z=0
  \qquad\text{or}\qquad
  2xZ\dot y+(1-Z)y=0 .
\label{eq:twobranches}
\end{equation}
The first integrates to $Z=1+kx$, that is $\psi=\text{const}$, whereupon
$\mathcal I^{(0)}_{d}=0$ reduces to a linear equation for $y$. For
$k\neq0$ its general solution is $y=A+B\sqrt{Z}$, the interior
Schwarzschild metric; for $k=0$, which is $\psi\equiv0$, the two
solutions coincide and the general solution is $y=A+B'x$, the branch previously noted. 
On the second branch $\dot y/y=-(1-Z)/2xZ$, and substituting this into
$\mathcal I^{(0)}_{d}=0$ gives, after clearing denominators,
\begin{equation}
  (Z-1)\bigl[(d-2)Z+1\bigr]\bigl(x\dot Z-Z+1\bigr)=0 .
\label{eq:branch2resid}
\end{equation}
The first factor gives $Z\equiv1$, that is $\psi\equiv0$; the
second gives the constant $Z=-1/(d-2)<0$ for $d>2$, excluded by
$Z>0$; the third returns $Z=1+kx$ and hence the first branch again.
\end{proof}

The physical content is a statement about how much a stellar interior can
reveal about the gravitational theory. By \eqref{eq:isoInd} the isotropy
condition is \emph{linear} in the weighted couplings
$\widetilde\alpha_{n}$: for a prescribed interior $(Z,y)$ the admissible
theories form the kernel of the linear map
\begin{equation}
  \bigl(\widetilde\alpha_{1},\ldots,\widetilde\alpha_{N}\bigr)
  \;\longmapsto\;
  \sum_{n=1}^{N}\widetilde\alpha_{n}\,f_{n}(x),
  \qquad
  f_{n}=n\,C^{n}\frac{(1-Z)^{n-2}}{x^{n}}\,\mathcal I_{n,d}[Z,y],
\label{eq:kernel}
\end{equation}
a linear subspace of coupling space. Theorem~\ref{thm:uniqueness} says
that this subspace is the whole of coupling space for the interior
Schwarzschild metric and for the affine branch $Z\equiv1$, $y=A+B'x$, and
for nothing else. On that branch an arbitrary $y$ requires in addition
$\alpha_{1}=0$, which is itself a proper subspace of the couplings. Every other interior therefore constrains the
couplings, and the constraint is linear: an observed or postulated
interior geometry that is not of constant density restricts the Lovelock
couplings to a proper subspace, while the isotropy condition of the
universal Schwarzschild-form constant-density interior places no
restriction on them at all. Theorem~\ref{thm:universal} therefore has a complementary implication.
Apart from the branch $\psi\equiv0$, the Schwarzschild-form
constant-density interior is the unique geometry whose pressure-isotropy
condition is satisfied without restricting the Lovelock couplings.  This
coupling independence concerns only the functional form of the geometry.
Equation~\eqref{eq:constW} determines the central sectional curvature
$\psi_{0}$ from the density through $W$, while
\eqref{eq:sigma} shows that the pressure profile and the
constant-density compactness condition remain coupling-dependent through
$\sigma$.  Universality of the geometry therefore does not imply
universality of its matter normalisation or stellar parameters.

\section{The critical dimensions}
\label{sec:critical}

The factors $d-1-2n$ and $d-2-2n$ carried by $V$ and $U$ in
\eqref{eq:VUexp} vanish at $d=2n+1$ and $d=2n+2$ respectively; we call
these the critical dimensions of the $n$-th Euler density. They are
ordinarily discussed by restricting the Lagrangian to a single order, but
no such restriction is needed. Which orders of a given theory can be
critical is settled by an elementary observation.

\begin{lemma}
\label{lem:onlyN}
In any dimension $d$, no Lovelock order $n<N=\lfloor(d-1)/2\rfloor$ is
critical.
\end{lemma}

\begin{proof}
If $n\le N-1$ then $2n+2\le2N$, while $d\ge2N+1$ by \eqref{eq:Nmax}.
Hence $d>2n+2>2n+1$, so neither critical relation can hold.
\end{proof}

The only order that can be critical in a given dimension is therefore the
maximal admissible one, $n=N$, and it is present precisely when the
theory is maximal.

\begin{proposition}
\label{prop:parity}
Let $\widetilde\alpha_{N}\neq0$, so that the theory is maximal and
$\deg W=N$. Then $d=2N+1$ if $d$ is odd and $d=2N+2$ if $d$ is even, and
correspondingly the leading coefficient of $V$ vanishes in odd dimensions
while that of $U$ vanishes in even dimensions:
\begin{equation}
  d\ \text{odd}:\ \deg V\le N-1,
  \qquad
  d\ \text{even}:\ \deg U\le N-1 .
\label{eq:degrees}
\end{equation}
If instead $n_{\max}<N$, the theory possesses no critical order at all.
\end{proposition}

The dichotomy is thus a property of maximal theories, and the hypothesis
cannot be dropped. Einstein gravity is maximal only in $d=3,4$ and has no
critical order in $d\ge5$; Einstein--Gauss--Bonnet is maximal in $d=5,6$
and has none in $d\ge7$; pure Lovelock of order $n$ is maximal exactly in
$d=2n+1$ and $d=2n+2$. Granted maximality, a Lovelock theory contains
exactly one critical order, namely its top one, whatever the remaining
couplings.

By \eqref{eq:closed} the consequence in odd dimensions is that the density
and radial pressure lose their $\psi^{N}$ terms, the highest order
reaching them only through the metric derivatives, since it survives in
$W_{\psi}$ but not in $V$; equivalently the coefficient $d-2n-1$ in
\eqref{eq:rhon}, \eqref{eq:prn} and \eqref{eq:Ind} vanishes at $n=N$. In
even dimensions the density and radial pressure retain their leading
terms and it is the tangential pressure that is depleted, through $U$.
\new{The depletion is partial: only the algebraic contribution of the top
order, the one entering $p_{t}$ through $U$, is lost, while its
derivative contributions through $W_{\psi}$, $W_{\psi\psi}$ and $V_{\psi}$
survive.}

A second consequence concerns the vacuum near the centre. Setting
$\rho=0$ in \eqref{eq:mass} gives the Wheeler polynomial
$W(\psi)=2\mu/r^{d-1}$, and as $r\to0$ the leading term of $W$ dominates,
so that
\begin{equation}
  1-Z=\psi r^{2}\simeq
  \Bigl(\frac{2\mu}{\widetilde\alpha_{N}}\Bigr)^{1/N}
  r^{\,2-(d-1)/N},
\label{eq:centre}
\end{equation}
an exponent equal to $0$ in odd dimensions and $-1/N$ in even ones. Hence
in odd dimensions $1-Z$ approaches the finite constant
$c=(2\mu/\widetilde\alpha_{N})^{1/N}$ at the centre. \new{Positivity of
$2\mu/\widetilde\alpha_{N}$ ensures a finite limiting angular factor but
not its interpretation: a solid-angle deficit requires $0<c<1$, while
$c<0$ gives a surplus, while $c\geq1$ gives $Z\leq0$ and is inadmissible
in the static region.} In even dimensions it diverges as
$r^{-1/N}$; the lower-order couplings affect only subleading behaviour.

\section{\texorpdfstring{The $V\equiv0$ obstruction and its removal}%
{The V=0 obstruction and its removal}}
\label{sec:nogo}

That no bound distribution of finite radius exists in pure Lovelock
gravity in $d=2N+1$ is known \cite{DadhichHansrajChilambwe2017}, the
companion obstruction to isothermality in the same dimension having been
established in \cite{DadhichHansrajMaharaj2016}. The closed formulation
identifies the exact algebraic origin of the first of these, and shows
that it is not a generic property of the critical dimension: the
fixed-sign mechanism behind it is equivalent to the condition
$V\equiv0$, which as a polynomial identity is $N$ independent conditions
on the coupling vector $(\alpha_{0},\ldots,\alpha_{N})$ and therefore
confines the theory to a line. Departing from that line removes the
mechanism, though it does not by itself establish that a star exists. Since the isotropy condition
\eqref{eq:totalderiv} is driven entirely by $V_{\psi}$, and $V$ also fixes
the radial pressure through $2p_{r}=aW_{\psi}-V$, it is natural to ask
when $V$ degenerates.

\begin{proposition}
\label{prop:Vflat}
Let $\widetilde\alpha_{N}\neq0$, so that the order $N$ is genuinely
active. In the critical odd dimension $d=2N+1$ the coefficient $d-1-2n$
vanishes
at $n=N$, so the highest Lovelock order drops out of $V$ altogether while
every lower order survives:
\begin{equation}
  V=\alpha_{0}+\sum_{n=1}^{N-1}\widetilde\alpha_{n}(d-1-2n)\,\psi^{n} .
\label{eq:Vcritical}
\end{equation}
Consequently $V_{\psi}\equiv0$ if and only if no order below $N$ is
active, that is, if and only if the theory is pure Lovelock of order
$N=(d-1)/2$ together with an arbitrary cosmological term; and $V\equiv0$
if and only if in addition $\alpha_{0}=0$.
\end{proposition}

When $V_{\psi}\equiv0$ the right-hand side of \eqref{eq:totalderiv}
vanishes and the isotropy condition integrates once,
\begin{equation}
  W_{\psi}\,\frac{\sqrt{Z}\,y'}{r}=K=\text{const},
\label{eq:firstintegral}
\end{equation}
so that it is solved by quadrature for an \emph{arbitrary} profile
$Z(r)$. Substituting into $2p_{r}=aW_{\psi}-V$ and using
$a=2Zy'/(ry)$ then gives the radial pressure as 
\begin{equation}
  p_{r}=K\,\frac{\sqrt{Z}}{y}-\frac{\alpha_{0}}{2}
\label{eq:prcritical}
\end{equation}
throughout the fluid, whatever the density profile.

\begin{theorem}[Obstruction]
\label{thm:nogo}
Let $\widetilde\alpha_{N}\neq0$, so that the leading coefficient
$(d-1-2N)\widetilde\alpha_{N}$ of $V$ forces $d=2N+1$ whenever
$V\equiv0$. If $V\equiv0$ then no static isotropic fluid sphere with
non-vanishing radial pressure possesses a pressure-free surface at finite
radius, for any density profile: either $p_{r}$ keeps a fixed non-zero sign
throughout, or $p_{r}\equiv0$. By
Proposition~\ref{prop:Vflat} this occurs precisely in pure Lovelock
gravity of order $N$ in the dimension $d=2N+1$ with vanishing
cosmological term.
\end{theorem}

\begin{proof}
With $V\equiv0$, \eqref{eq:prcritical} gives $p_{r}=K\sqrt{Z}/y$. In a
regular static interior $Z>0$ and $y>0$, so $p_{r}$ carries the fixed sign
of $K$ and can vanish at some $r=R$ only if $K=0$, in which case
$p_{r}\equiv0$ throughout.
\end{proof}

\new{%
The degenerate alternative is not vacuous, and it is worth saying what it
is. If $K=0$ then $aW_{\psi}=2p_{r}=0$, so on a non-degenerate branch
$a=0$ and the temporal potential is constant; but $V\equiv0$ leaves
$2\rho=-bW_{\psi}$ untouched, and $b$ need not vanish. What the critical
pure theory excludes is therefore a bounded configuration supported by
pressure, not a non-zero density: a static pressureless distribution
survives, with no gravitational redshift gradient to balance. The
obstruction is to the star, not to the matter.}

The known non-existence theorem is thus recovered as the special case
$V\equiv0$, and is strengthened in three respects: the condition is shown
to be not merely sufficient but exactly characterising \emph{for this
fixed-sign obstruction}, the narrowness of the locus it defines in
coupling space is quantified, and the mechanism is exhibited as the fixed
sign of \eqref{eq:prcritical} rather than as a property of any particular
density profile, so that it holds for every profile at once. \new{The
converse is weaker than the theorem: departing from $V\equiv0$ removes
this obstruction, but does not by itself guarantee that a regular
interior matched to a vacuum exterior exists.} What the formulation makes plain is how narrow the
condition is. Two minimal departures from it restore the possibility of a
finite pressure-free surface, and for each we exhibit an explicit
finite pressure-supported interior possessing one. We call these interiors rather than stars: the exterior
matching is governed by the Lovelock surface tensor
\cite{Davis2003,Gravanis2003}, which is nonlinear in the jump of the extrinsic curvature which analysis is not carried out here.

\subsection*{Lifting the obstruction: lower-order Lovelock terms}
Any active order below $N$ contributes to $V$ and destroys the first
integral. The simplest illustration is $d=5$, where $N=2$ and the
Gauss--Bonnet term cancels from $V$ but the Einstein term does not:
\begin{equation}
  V\big|_{d=5}=2\widetilde\alpha_{1}\psi ,
  \qquad
  V_{\psi}\big|_{d=5}=2\widetilde\alpha_{1}.
\label{eq:V5}
\end{equation}
Take the uniform-density interior $Z=1-\psi_{0}r^{2}$,
$y=A-B\sqrt{Z}$ of Theorem~\ref{thm:universal}. Then
\begin{equation}
  p_{r}=\frac{\psi_{0}\bigl[2B(\widetilde\alpha_{1}
   +\widetilde\alpha_{2}\psi_{0})\sqrt{Z}
   -A\widetilde\alpha_{1}\bigr]}{A-B\sqrt{Z}},
\label{eq:pr5EGB}
\end{equation}
which vanishes at the radius where
$\sqrt{Z_{R}}=A\widetilde\alpha_{1}/
\bigl[2B(\widetilde\alpha_{1}+\widetilde\alpha_{2}\psi_{0})\bigr]$: an
admissible surface exists. Choosing $\psi_{0}=1$,
$\widetilde\alpha_{1}=3$, $\widetilde\alpha_{2}=3/5$, $B=1$ and
$A=1.69706$ places the surface at $Z_{R}=1/2$, that is $R=0.7071$, with
$p_{r}(0)=3.0253$ decreasing monotonically to zero and $y>0$ throughout.
Setting $\widetilde\alpha_{1}=0$ in \eqref{eq:pr5EGB} returns
$p_{r}=2B\widetilde\alpha_{2}\psi_{0}^{2}\sqrt{Z}/(A-B\sqrt{Z})$, which by
Theorem~\ref{thm:nogo} has no zero in the regular static region $Z>0$;
its only zero is at $Z=0$, which is not an admissible stellar surface. The
same metric, in the same
dimension, is a finite pressure-supported interior with a pressure-free
surface in Einstein--Gauss--Bonnet gravity, but not in pure Gauss--Bonnet
gravity.

\subsection*{Lifting the obstruction: a cosmological term}
Alternatively $V$ may be made non-zero while keeping the theory pure, by
retaining $\alpha_{0}$. Equation \eqref{eq:prcritical} then permits a
surface, and the quadrature \eqref{eq:firstintegral} generates a solution
for any chosen $Z$. Taking $d=5$, $N=2$ and
\begin{equation}
  Z=\frac{1}{1+cr^{2}},
  \qquad
  y=A+B\bigl(1+cr^{2}\bigr)^{5/2},
\label{eq:exactsol}
\end{equation}
whose sectional curvature $\psi=c/(1+cr^{2})$ is not constant, one finds
$K=10B\widetilde\alpha_{2}c^{2}$ and
\begin{equation}
  \rho=\frac{\alpha_{0}}{2}
   +\frac{2\widetilde\alpha_{2}c^{2}}{(1+cr^{2})^{3}},
  \qquad
  p_{r}=\frac{10B\widetilde\alpha_{2}c^{2}}
   {\sqrt{1+cr^{2}}\;\bigl[A+B(1+cr^{2})^{5/2}\bigr]}-\frac{\alpha_{0}}{2}.
\label{eq:exactrhop}
\end{equation}
The density is positive, finite at the centre and monotonically
decreasing, with
$\rho'=-12\widetilde\alpha_{2}c^{3}r/(1+cr^{2})^{4}$; the pressure
decreases from $p_{r}(0)=10B\widetilde\alpha_{2}c^{2}/(A+B)-\alpha_{0}/2$
towards $-\alpha_{0}/2$ and so vanishes at exactly one finite radius
whenever $\alpha_{0}>0$ and
$10B\widetilde\alpha_{2}c^{2}>\tfrac12\alpha_{0}(A+B)$. With
$\widetilde\alpha_{2}=c=B=1$, $A=40$ and $\alpha_{0}=1/5$ the surface lies
at $R\simeq1.447$, the central density is $\rho(0)=2.1$, the central
pressure $p_{r}(0)\simeq0.144$ and the surface compactness
$1-Z(R)\simeq0.677$, with $p_{r}\le\rho$ throughout and maximum squared
sound speed $dp_{r}/d\rho\simeq0.611$: the model is causal and satisfies
the dominant energy condition.

The picture is therefore not that critical odd dimensions forbid stars,
but that the obstruction is an accident of the pure theory with no
cosmological term, confined to a line in the $(N+1)$-dimensional coupling
space and codimension one only within the pure-$N$ family that a
cosmological term augments. It arises because $V$, which by
\eqref{eq:Vcritical} has lost its leading coefficient, has nothing else
left; restoring any lower order, or any cosmological term, removes it.

\section{Stellar sequences}
\label{sec:sequences}

The closed forms of Theorem~\ref{thm:closed} reduce the construction of
Lovelock stellar models to the integration of two ordinary differential
equations, at any order and in any dimension, without the order-by-order
assembly of field equations that has hitherto confined such calculations
to the lowest cases. Eliminating $a$ from $2p_{r}=aW_{\psi}-V$ and
substituting into the conservation equation \eqref{eq:conservation} gives
the Lovelock analogue of the Tolman--Oppenheimer--Volkoff system,
\begin{equation}
  \frac{dm}{dr}=\rho\,r^{d-2},
  \qquad
  \frac{dp_{r}}{dr}
  =-\frac{r}{2Z}\,\frac{2p_{r}+V}{W_{\psi}}\,\bigl(\rho+p_{r}\bigr),
\label{eq:TOV}
\end{equation}
in which $\psi$ is obtained at each radius by inverting
$W(\psi)=2m/r^{d-1}$ on the branch continuously connected to the Einstein
limit, and $Z=1-\psi r^{2}$. All dependence on the order and the
dimension is carried by the polynomial $W$ and its companion $V$;
specialising to a given theory is a matter of choosing coefficients. For
$d=4$ and $W=2\psi$ the system reduces to the standard
Tolman--Oppenheimer--Volkoff equations.

We integrate \eqref{eq:TOV} for the linear equation of state
$p_{r}=(\rho-4B)/3$, with the scale fixed by $B=1$, from a regular centre
outwards to the surface $p_{r}(R)=0$; varying the central pressure
generates a one-parameter sequence. In five and six dimensions the
procedure reproduces the Einstein--Gauss--Bonnet sequences already
familiar from the literature, and we use these as a check. The interest
of the present formalism, however, is that nothing changes when the order
is raised: we therefore take as our principal illustration the two
dimensions distinguished by string theory, $d=10$ and $d=11$, in each of
which \emph{all admissible} Lovelock orders are retained.

These two are not an arbitrary choice. By \eqref{eq:Nmax} the highest
admissible order is $N=4$ in ten dimensions and $N=5$ in eleven, so that
\begin{equation}
  d=10=2N+2,
  \qquad
  d=11=2N+1,
\label{eq:stringparity}
\end{equation}
and the superstring and M-theory dimensions realise precisely the two
cases of Proposition~\ref{prop:parity}: in $d=10$ the leading coefficient
of $U$ vanishes, in $d=11$ that of $V$. The eleven-dimensional case is
thus critical in the sense of Proposition~\ref{prop:parity}: the
fifth-order contribution is critical there, but the obstruction does not apply to the full Lovelock theory, since the active lower-order terms ensure that \(V\not\equiv0\).
Bounded configurations accordingly exist there in large numbers, the
obstruction of Theorem~\ref{thm:nogo} applying only to the pure theory.

We take couplings governed by a single length scale, applied to the
\emph{weighted} couplings since it is these and not the bare ones that
enter the field equations:
\new{Throughout the numerical work $\widetilde\alpha_{1}=d-2$, that is
$\alpha_{1}=1$, so that the Einstein term is canonically normalised, and
the tabulated values of $\kappa$ are in the units fixed by $B=1$ in the
equation of state.}
\begin{equation}
  \widetilde\alpha_{n}=\widetilde\alpha_{1}\,\kappa^{\,n-1},
  \qquad
  W(\psi)=\widetilde\alpha_{1}\sum_{n=1}^{N}\kappa^{\,n-1}\psi^{n}
  =\widetilde\alpha_{1}\,\psi\,\frac{1-(\kappa\psi)^{N}}{1-\kappa\psi},
\label{eq:tower}
\end{equation}
the quotient form being a convenience whose apparent pole at
$\kappa\psi=1$ is removable since the denominator is a factor of the numerator. Here $\kappa$ has dimensions of area,
$\kappa\psi$ is dimensionless, and $\kappa=0$ recovers Einstein gravity
while increasing $\kappa$ switches on four or five curvature orders
together. We refer to \eqref{eq:tower} as a \emph{one-scale maximal
tower}: all admissible orders are active, in the sense of
Section~\ref{sec:weights}, but the couplings traverse a single
one-parameter slice through the four- or five-dimensional coupling space
rather than exploring it generically.

The interpretation of these calculations requires two qualifications. Ten and eleven dimensions are
used here because they are distinguished by superstring theory and
M-theory and because they realise the two parity cases, not because
\eqref{eq:tower} is derived from any string effective action; we make no
such identification. Similarly the equation of state
$p_{r}=s(\rho-\rho_{s})$ with $s=1/3$ and $\rho_{s}=4B$ is adopted as a
phenomenological self-bound probe, chosen to match earlier Gauss--Bonnet
stellar calculations rather than because the coefficient $1/3$ is natural
for conformal matter in ten dimensions; a robustness check against the
dimension-dependent slope $s=1/(d-1)$ is reported below. Finally, the
mass plotted is
\begin{equation}
  M=m(R)=\tfrac12 R^{\,d-1}W\bigl(\psi(R)\bigr),
\label{eq:Mdef}
\end{equation}
which carries dimensions of $(\text{length})^{d-3}$; setting $B=1$ fixes
the scale, and only ratios within a single dimension are meaningful.
Central pressures run over $6\le p_{c}\le 500$.
Figures~\ref{fig:MRa}--\ref{fig:MRc} show the resulting sequences.

\begin{figure}[htbp]
\centering
\includegraphics[width=0.62\textwidth]{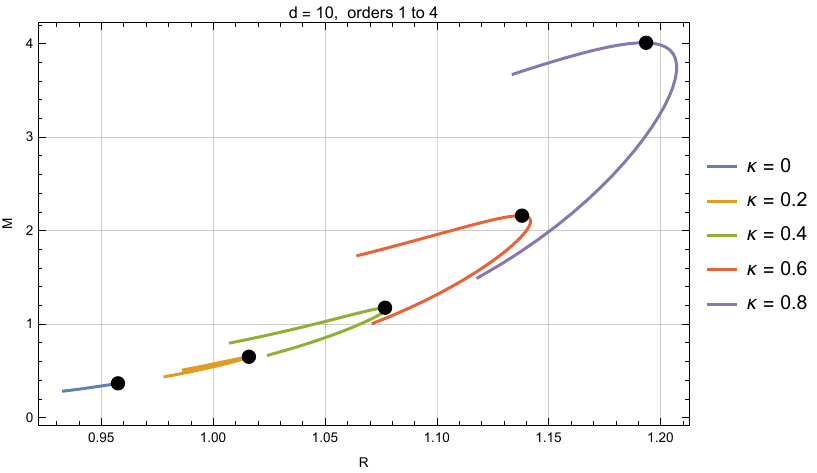}
\caption{Mass--radius sequences in ten dimensions with all four
admissible Lovelock orders active, obtained from \eqref{eq:TOV} for
$p_{r}=(\rho-4B)/3$ with $B=1$ and the geometric tower
\eqref{eq:tower}. Filled circles mark the maximum mass. Since
$d=10=2N+2$, this is the case in which the leading coefficient of $U$
vanishes.}
\label{fig:MRa}
\end{figure}

\begin{figure}[htbp]
\centering
\includegraphics[width=0.62\textwidth]{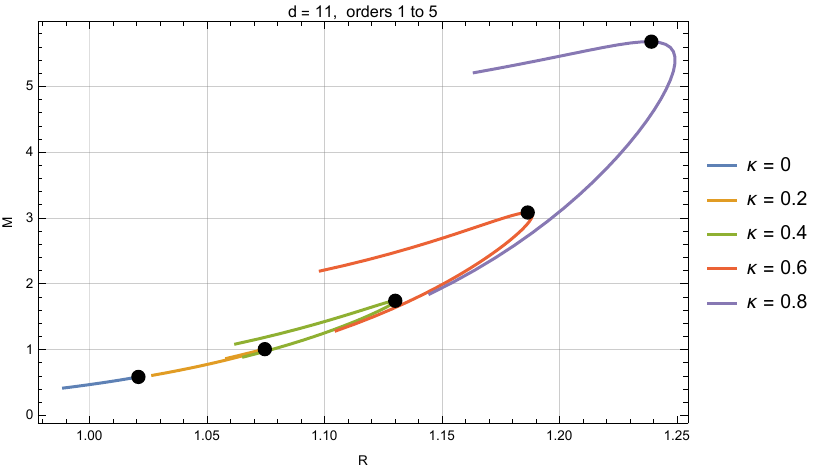}
\caption{As Figure~\ref{fig:MRa}, in eleven dimensions with all five
admissible Lovelock orders active. Here $d=11=2N+1$, the case in which
the leading coefficient of $V$ vanishes; bounded configurations
nevertheless exist, because the lower orders keep $V\not\equiv0$.}
\label{fig:MRb}
\end{figure}

\begin{figure}[htbp]
\centering
\includegraphics[width=0.62\textwidth]{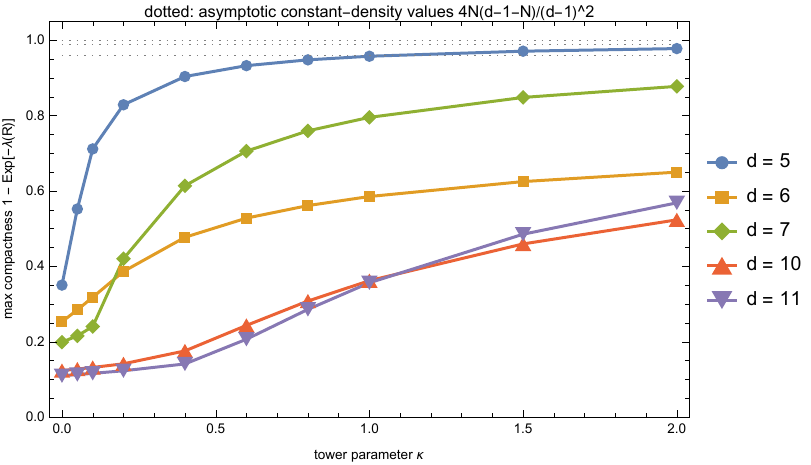}
\caption{The largest compactness $1-e^{-\lambda(R)}$ reached along each
sequence, as a function of the tower parameter $\kappa$, in
$d=5,6,7,10,11$; mass--radius curves are shown above for the two highest
of these only. Dotted lines are the asymptotic constant-density
values $4N(d-1-N)/(d-1)^{2}$ for the maximal tower, equivalent to the
high-density limit obtained in
\cite{Bueno2026,DadhichChakraborty2017}; these apply for $\kappa>0$,
the isolated endpoint $\kappa=0$ being Einstein gravity, for which the
corresponding value is $4(d-2)/(d-1)^{2}$. All sequences computed here
remain below the applicable value.}
\label{fig:MRc}
\end{figure}

We can make a number of  deductions from the plots. 
\begin{itemize}
\item[(i)] 
{\it Persistence of the general-relativistic topology.} Every
sequence computed --- Einstein, Einstein--Gauss--Bonnet, and the complete
four- and five-term theories in ten and eleven dimensions --- displays the
familiar structure of a rising branch, a turning point at a maximum mass,
and a descending branch. Adding curvature orders deforms the curve but
does not alter its topology. Whatever else changes, the qualitative
organisation of relativistic stellar structure survives intact from
Einstein gravity to the full Lovelock hierarchy in the dimensions where
superstring theory and M-theory are formulated.

\item[(ii)] {\it For the positive one-scale tower considered here, higher orders make stars heavier and more compact.} At fixed
equation of state the maximum mass rises steeply with $\kappa$, and so
does the compactness of the configuration at which it is attained, as
Table~\ref{tab:maxmass} records.

\begin{table}[htbp]
\centering
\begin{tabular}{lcc|cc|cc}
\hline
& \multicolumn{2}{c|}{$M_{\max}$}
& \multicolumn{2}{c|}{$R$ at $M_{\max}$}
& \multicolumn{2}{c}{compactness at $M_{\max}$}\\
$\kappa$ & $d=10$ & $d=11$ & $d=10$ & $d=11$ & $d=10$ & $d=11$\\
\hline
$0$   & $0.364$ & $0.583$ & $0.958$ & $1.021$ & $0.123$ & $0.110$\\
$0.2$ & $0.648$ & $1.005$ & $1.016$ & $1.075$ & $0.141$ & $0.123$\\
$0.4$ & $1.173$ & $1.740$ & $1.077$ & $1.130$ & $0.165$ & $0.139$\\
$0.6$ & $2.158$ & $3.081$ & $1.137$ & $1.187$ & $0.199$ & $0.162$\\
$0.8$ & $4.008$ & $5.677$ & $1.193$ & $1.238$ & $0.251$ & $0.204$\\
$1.5$ & $24.98$ & $52.60$ & $1.384$ & $1.393$ & $0.430$ & $0.507$\\
\hline
\end{tabular}
\caption{Maximum mass along each sequence, the radius at which it occurs,
and the compactness $1-e^{-\lambda(R)}$ \emph{of that same
maximum-mass configuration}. This is not the largest compactness reached
anywhere along the sequence, which is the quantity plotted in
Figure~\ref{fig:MRc} and is in general attained at a different central
pressure.}
\label{tab:maxmass}
\end{table}

\noindent
The mass carries dimensions of $(\text{length})^{d-3}$, so the numerical
magnitudes are a matter of units and only ratios within a column are
meaningful. Relative to the Einstein value the maximum mass is raised by
factors of $1.8$, $3.2$, $5.9$, $11$ and $69$ in ten dimensions as
$\kappa$ runs through the table, and by $1.7$, $3.0$, $5.3$, $9.7$ and
$90$ in eleven. The rising branches of
Figures~\ref{fig:MRa} and~\ref{fig:MRb} coincide at low central pressure and separate only
near the turning point: the higher-curvature terms are dormant in weakly
curved configurations and switch on only for compact ones, which is the
expected behaviour of a curvature expansion and a useful internal
consistency check.

\item[(iii)] {\it Compactness stays well below the constant-density value.}
Figure~\ref{fig:MRc} shows the largest compactness attained along each
sequence, together with
\begin{equation}
  \mathcal C_{\rm cd}=\frac{4N(d-1-N)}{(d-1)^{2}},
\label{eq:cdvalue}
\end{equation}
the value obtained by letting $\psi\to\infty$ in the effective order
$2\psi W_{\psi}/(d-1)W$ for a constant-density configuration, with $N$
the highest active order. This is the maximal tower for every
$\kappa>0$; at the isolated endpoint $\kappa=0$ all higher couplings
vanish, $\deg W=1$, and the applicable value is the Einstein one,
$4(d-2)/(d-1)^{2}$. Equation \eqref{eq:cdvalue} is the
high-density constant-density limit of \cite{Bueno2026}, recovered here
in the notation of Section~\ref{sec:uniform}, and we use it only as a
reference line. We call
this the asymptotic constant-density value rather than a bound: it is not
here proved to constrain every regular fluid satisfying our assumptions,
and in $d=11$ it equals unity, which is the horizon limit and therefore no
restriction at all. What is observed is that all sequences computed
remain below it, and by a wide margin: with $\mathcal C_{\rm cd}=80/81$ in
ten dimensions and $1$ in eleven, the computed maxima reach only $0.372$
and $0.377$ at $\kappa=1$, and $0.533$ and $0.588$ at $\kappa=2$. The
self-bound equation of state is evidently far from realising the
incompressible limit, which is the expected behaviour and a useful
consistency check on the numerics.

\item[(iv)] {\it Sensitivity.} The dependence is steep enough that stellar
structure is a sharp probe of the couplings, and in the parametrisation
\eqref{eq:tower} the measure of that sensitivity is approximately
dimension-insensitive across the five cases examined. Writing the Lovelock length as $\sqrt{\kappa}$,
the maximum mass departs by ten per cent from its Einstein value at
\begin{equation}
  \frac{\sqrt{\kappa}}{R}\simeq
  0.196,\;0.202,\;0.201,\;0.191,\;0.182
  \qquad\text{in } d=5,6,7,10,11,
\label{eq:sensitivity}
\end{equation}
that is, at a Lovelock length close to one fifth of the stellar radius in
every case across the five examined, whatever the number of active
curvature orders. The
number of orders governs how fast the mass then grows, but not the scale
at which the growth begins.

\item[(v)] {\it Robustness.} Two checks bear on the origin of the pattern.
The first asks whether the enhancement is due to the highest order alone;
it does not test a different coupling tower, which the second check does
not do either.Switching the orders on successively at fixed $\kappa=0.8$ within the
same hierarchy gives the results in
Table~\ref{tab:successive-orders}. The effect is cumulative rather than
being driven by the highest order alone, with each additional curvature
order contributing a diminishing increment.

\begin{table}[t]
\centering
\begin{tabular}{lccccc}
\hline
& \multicolumn{5}{c}{$M_{\max}$ with orders $1,\ldots,N$ active}\\
$N$ & $1$ & $2$ & $3$ & $4$ & $5$\\
\hline
$d=10$ & $0.364$ & $1.810$ & $3.033$ & $4.008$ & ---\\
$d=11$ & $0.583$ & $2.700$ & $4.190$ & $5.078$ & $5.677$\\
\hline
\end{tabular}
\caption{Maximum mass at fixed $\kappa=0.8$ as the admissible Lovelock
orders are activated successively. Each column includes all orders from
$1$ through $N$; fifth order is not admissible in $d=10$.}
\label{tab:successive-orders}
\end{table}
 Second, and separately, replacing the slope
$s=1/3$ by the dimension-dependent value $s=1/(d-1)$ changes the mass
scale considerably but not the qualitative conclusion: in ten dimensions
$M_{\max}$ rises from $0.0075$ at $\kappa=0$ to $0.691$ at $\kappa=0.8$, a
factor of $92$ against the factor of $11$ found for $s=1/3$, while the
compactness rises from $0.054$ to $0.366$ rather than from $0.123$ to
$0.251$. The enhancement of mass and compactness by the higher orders is
if anything stronger for the softer equation of state.

\item[(vi)] {\it Scope.} These are intrinsically higher-dimensional sequences,
and are not directly comparable with four-dimensional neutron-star
observations without an additional dimensional-reduction prescription. We
therefore present them as theoretical demonstrations of the
arbitrary-order system \eqref{eq:TOV} rather than as phenomenological
fits.

\end{itemize}

As a validation of the numerical implementation, the code reproduces the four-dimensional general-relativistic sequence for the same equation of state.  Sequence maxima are determined by golden-section refinement in
$\log p_{c}$ rather than by selection from a fixed grid.

\section{Discussion}

The central analytical result of this paper is that the order-by-order
Lovelock sums can be absorbed into the characteristic polynomials
$W$, $V$ and $U$ for the complete anisotropic matter tensor, not merely
for its perfect-fluid specialisation.
Equation \eqref{eq:closed} contains no $\Sigma$: the couplings and the
whole Lovelock hierarchy are carried by $W$, and the dimension by the
integer coefficients of the operation $P\mapsto kP-2\psi P_{\psi}$ which
generates $V$ and $U$ from it. \new{For a perfect fluid a reduction of this
kind was already available \cite{Bueno2026}, and \eqref{eq:totalderiv} is
its isotropy equation in other variables; what \eqref{eq:closed} supplies
is the tangential pressure, and with it the anisotropy \eqref{eq:aniso},
which the perfect-fluid assumption renders redundant and therefore
conceals.} A calculation performed once holds for every order
and every dimension at once, and specialising to a named theory becomes a
matter of choosing a polynomial rather than of redoing a derivation. That
this is possible at all rests on the counting of
Section~\ref{sec:direct}: the series truncates, and the system closes on a
single polynomial, because the orbit spheres form a constant-curvature
block with only $d-2$ directions to antisymmetrise over. The recurrences
of Section~\ref{sec:induction} confirm the same coefficients
independently, and the radial projection follows from the symmetry of the
block structure rather than by assertion.

The structural consequences depend on the two derived polynomials. Because
$\mathcal I_{n,d}$ is affine in $n$, a metric solving the isotropy
condition for arbitrary couplings must annihilate two independent
combinations, and this forces it to be the interior Schwarzschild metric:
constant density is, apart from the branch of vanishing orbit sectional
curvature, the unique  interior irrespective of theory, and Einstein
together with Gauss--Bonnet already suffices to force it. Since the
isotropy condition is linear in the weighted couplings, every other
interior confines them to a proper subspace, so the isotropy condition of
a constant-density star cannot distinguish the higher-curvature couplings
while that of any other interior can; the matter normalisation of such a
star does depend on them, through the sectional curvature that the
density selects and through the compactness bound. The polynomial $V$ likewise acquires a direct
physical meaning: the fixed-sign mechanism behind the known non-existence
of bounded pure Lovelock spheres in $d=2N+1$ is traced precisely to
$V\equiv0$, which holds only on a single line in coupling space, and
either a lower curvature order or a cosmological term removes that
obstruction. At uniform density the identity $V\equiv0$ produces the
pressureless limit of the quasi-topological analysis \cite{Bueno2026};
retaining the Einstein term excludes that identity, although a pointwise
root $V(\psi_{0})=0$ remains possible for couplings of indefinite sign. Separately, $V$
and $U$ carry the factors that mark the critical dimensions, only the
maximal admissible order can ever be critical, and the parity of $d$
decides which of the two degenerates.

The numerical section is a demonstration that these statements are
computable, not merely formal. The system \eqref{eq:TOV} integrates as
readily with five active curvature orders as with one, and in ten and
eleven dimensions --- the two parity cases --- the general-relativistic
branch topology of a rising sequence, a maximum mass and a descending
branch survives intact in the sequences computed here. The eleven-dimensional sequences are the sharpest
illustration of Section~\ref{sec:nogo}: the fifth-order contribution is
critical there by Proposition~\ref{prop:parity}, but the tower as a whole
is not on the no-go locus $V\equiv0$, since the
lower orders keep $V\not\equiv0$, and bounded stars exist in profusion. For the positive one-scale tower adopted there, higher
orders raise the maximum mass and the compactness substantially and
cumulatively, an effect robust against changing the equation-of-state
slope, though the compactness remains far below the asymptotic
constant-density value. Nothing in the closed equations forces this
behaviour for couplings of arbitrary sign; it is a property of the tower
integrated, not of the hierarchy as such.

Natural extensions include the junction problem touched on in in earlier sections, which requires the Lovelock surface tensor evaluated for the polynomial $W$; the classification of the degenerate branches $W_{\psi}=0$, on which several statements above fail together; the perturbation problem, whose resolution would be needed before any tidal or stability analysis at general order; and the systematic use of \eqref{eq:isoInd} as a generating equation for exact interiors across the
Lovelock family, for which its linearity in $y$ at prescribed $Z$ is the natural point of attack.

\section*{Acknowledgements}

The author acknowledges the encouragement and direction by the late Prof. Naresh Dadhich (IUCAA) who motivated him to search for  closed form summation-free expressions for the Lovelock equations over a decade ago.

\end{document}